%% file: pm-main.tex
\documentclass[11pt,letterpaper]{article}

\input{pm-defs}

\begin{document}

\title{Polymatroid Prophet Inequalities}

\author{%
	Paul D\"{u}tting\thanks{%
	Department of Computer Science, Cornell University, 136 Hoy Road, Ithaca, NY 14850, USA. Email: \texttt{paul.duetting@cornell.edu}. Research supported by an SNF Postdoctoral Fellowship.}
	\and 
	Robert Kleinberg\thanks{%
	Department of Computer Science, Cornell University, 124 Hoy Road, Ithaca, NY 14850, USA. Email: \texttt{rdk@cs.cornell.edu}. Supported in part by NSF award AF-0910940, AFOSR grant FA9550-09-1-0100, a Microsoft Research New Faculty Fellowship, and a Google Research Grant.}
}

\date{}

\maketitle

\input{pm-abs}
\input{pm-intro}
\input{pm-prelims}
\input{pm-alg}

\input{pm-md}

\bibliographystyle{abbrvnat}
\bibliography{abb,polymatroid}

\appendix

\input{pm-appendix}

\end{document}

%% file: pm-defs.tex
\usepackage{hyperref}
\usepackage[numbers]{natbib}
\usepackage{xspace}
\usepackage{xcolor}
\usepackage{graphicx}
\usepackage{amsmath,amssymb,amsthm}
\usepackage[margin=1in]{geometry}
\usepackage{tikz}

\theoremstyle{plain}
\newtheorem{proposition}{Proposition}
\newtheorem{theorem}{Theorem}
\newtheorem{lemma}{Lemma}
\newtheorem{corollary}{Corollary}

\theoremstyle{definition}

\newtheorem{definition}{Definition}

\newcommand{\secref}[1]{Section~\ref{#1}}
\newcommand{\thmref}[1]{Theorem~\ref{#1}}

\newcommand{\proref}[1]{Proposition~\ref{#1}}
\newcommand{\figref}[1]{Figure~\ref{#1}}

\newcommand{\appref}[1]{Appendix~\ref{#1}}
\newcommand{\corref}[1]{Corollary~\ref{#1}}

\newcommand{\ie}{i.e.,\xspace}

\newcommand{\E}{\mathbb{E}}
\newcommand{\M}{\mathcal{M}}
\newcommand{\U}{\mathcal{U}}
\newcommand{\I}{\mathcal{I}}

\newcommand{\reals}{\mathbb{R}}
\newcommand{\Pf}{P_f}
\newcommand{\midd}{:}
\newcommand{\OPT}{\text{OPT}}

\newcommand{\naturals}{\mathbb{N}}
\newcommand{\F}{\mathcal{F}}
\newcommand{\cvec}[1]{{\mathbf{q}(#1)}}
\newcommand{\card}[2]{{q_{#1}(#2)}}

\newcommand{\wpm}{{\tau}}

%% file: pm-abs.tex
\begin{abstract}
Consider a gambler and a prophet who observe a sequence of independent, non-negative numbers. The gambler sees the numbers one-by-one whereas the prophet sees the entire sequence at once. The goal of both is to decide on fractions of each number they want to keep so as to maximize the weighted fractional sum of the numbers chosen.

The classic result of Krengel and Sucheston (1977-78) asserts that if both the gambler and the prophet can pick one number, then the gambler can do at least half as well as the prophet. Recently, Kleinberg and Weinberg (2012) have generalized this result to settings where the numbers that can be chosen are subject to a matroid constraint.

In this note we go one step further and show that the bound carries over to settings where the fractions that can be chosen are subject to a polymatroid constraint. This bound is tight as it is already tight for the simple setting where the gambler and the prophet can pick only one number. An interesting application of our result is in mechanism design, where it leads to improved results for various problems.
\end{abstract}

%% file: pm-intro.tex
\section{Introduction} \label{sec:intro}

Prophet inequalities compare the performance of an online algorithm to the optimum offline algorithm in settings that involve making selections from a sequence of random elements. The online algorithm knows the distribution from which the elements will be sampled, while the optimum offline algorithm knows the sequence of sampled elements. Prophet inequalities thus bound the relative power of online and offline algorithms in Bayesian settings. Not surprisingly, they play an important role in the analysis of online and offline algorithms in these settings. A slightly less obvious application is in algorithmic mechanism design, where they are used to design simple yet approximately optimal mechanisms. 

A classic result of \citet{KrSu77,KrSu78} shows that when both the online algorithm and the offline algorithm get to pick exactly one number, then the online algorithm can do at least half as well as the offline algorithm. More formally, if $w_1,\dots,w_n$ is a sequence of independent, non-negative, real-valued random variables satisfying $\E[\max_i w_i] < \infty$, then there exists a stopping rule $\tau$ such that
\[
	\E[w_\tau] \ge \frac{1}{2} \cdot \E[\max_i w_i].
\]
This bound is, for example, achieved by an elegant algorithm of \citet{SaCa84}. This algorithm chooses a threshold $T$ such that $\Pr(\max_i w_i > T) = \frac{1}{2}$, and selects the first element whose weight exceeds this threshold. Alternatively, as described by \citet{KlWe12}, this bound can be obtained by choosing threshold $T = \E[\max_i X_i]/2$ and picking the first element whose weight exceeds the threshold.

\citet{KlWe12} recently extended this result to matroid settings. In a matroid setting we are given a ground set $\U$ and a non-empty downward-closed family of independent sets $\I \subseteq 2^{|\U|}$ satisfying the exchange axiom: for all pairs of sets $I, J \in \I$ and $|I| < |J|$ there exists an element $j \in J$ such that $I \cup \{x\} \in \I$. For these settings they prove that if both the online and the offline algorithm have to pick an independent set of numbers, then the online algorithm again can do at least half as well as the offline algorithm. More formally, if $w_1,\dots,w_n$ is a sequence of independent, non-negative, real-valued random variables satisfying $\E[\max_i w_i] < \infty$, then there is a way to pick $A \in \I$ in an online fashion such that
\[
	\E\left[\sum_{i \in A} w_i\right] \ge \frac{1}{2} \cdot \E\left[\max_{B \in \I} \sum_{i \in B} w_i\right].
\]

An important application of the original result of \citet{KrSu77,KrSu78} and its generalization by \citet{KlWe12} is in algorithmic mechanism design, where --- as was first observed by \citet{ChHa10} --- prophet inequalities can be used to prove performance guarantees for simple, truthful mechanisms based on sequential posted pricing.

\paragraph{Our Contribution}
We extend the previous results to polymatroid settings. In a polymatroid setting we are given a ground set $\U$ and a submodular\footnote{A set function $f$ is submodular if for all $X \subset Y \subseteq \U$, $f(X \cup Y) + f(X \cap Y) \le f(X) + f(Y)$.} set function $f : 2^{\U} \to \reals$. A vector $x = (x_1,\dots,x_n)$ is feasible if $x \in \Pf = \{x \mid \sum_{i \in S} x_i \le f(S) \text{ for all } S \subseteq \U\}$.  We will restrict ourselves to integer-valued set functions for ease of exposition; our results trivially extend to rational-valued functions by scaling. For this setting we prove that if the goal of both the online and the offline algorithm is to maximize $w \cdot x$ over feasible $x$ and $w = (w_1,\dots,w_n)$ are the elements of the random sequence, then the online algorithm can again do at least half as well as the offline algorithm. More formally, if $w_1,\dots,w_n$ is a sequence of independent, non-negative, real-valued random variables satisfying $\E[\max_i w_i] < \infty$, then there exists a way to choose a feasible $x$ in an online fashion (\ie choosing $x_i$ when $w_1,\ldots,w_i$ have been revealed but $w_{i+1},\ldots,w_n$ have not yet been revealed) such that
\[
	\E\left[w \cdot x\right] \ge \frac{1}{2} \cdot \E\left[\max_{y \in \Pf} w \cdot y\right].
\]
We prove this result by reducing the polymatroid setting with independent weights to a matroid setting with limited correlation between weights. Specifically, we transform an input sequence to the polymatroid problem into an input sequence to the matroid problem by repeating the (element, weight) pairs in the input sequence to the polymatroid problem. We show that this leads to a matroid consisting of several blocks, each corresponding to an element of the ground set of the polymatroid. The resulting distributions of weights have the property that each element of a block is associated with a weight that is independent from the weights of elements of other blocks, but that is the same for all elements of a block. We call the resulting matroid a block-structured matroid, and the resulting distributions block-structured distributions. Our main technical contribution apart from the reduction itself is to prove that the Kleinberg-Weinberg algorithm, although originally developed for the independent weights case, also applies to block-structured matroids and block-structured distributions. This result is not only the main building block of our result for polymatroids, but it is also interesting in its own right as it constitutes a prophet inequality for a setting in which the weights can be correlated. Such inequalities are rare in the literature, and those that have appeared in the past required either negative dependence~\cite{Bo91,RiSaCa87}, martingale moment sequences~\cite{CoKe86}, or additive rather than multiplicative bounds on the difference between the online and offline algorithm's expected payoffs~\cite{HiKe83}. Our prophet inequality allows for a different --- though, unfortunately, still very stringent --- restriction on the type of correlation allowed.

An important implication of our result are novel approximation results for mechanism design problems with polymatroid structure. This class of problems comprises, amongst others, position auctions \cite{GoMi12} and spatial markets \cite{BaNi09}.

\paragraph{Related Work}
We have already described the result by \citet{KrSu77,KrSu78} for the case in which both the online algorithm and the offline algorithm are allowed to pick one number, showing that the online algorithm can do at least half as well as the offline algorithm. This bound is tight. This result has been extended to the case where both the online algorithm and the offline can pick $k$ numbers by \citet{Al11}, showing that the online-to-offline ratio is at most $1-1/(\sqrt{k+3}).$ This matches the aforementioned tight bound when $k=1$, and it remains nearly tight for $k>1$, in the sense that a ratio of $1 - o(1/\sqrt{k})$ is known to be unattainable. Finally, as already mentioned, \citet{KlWe12} have extended the bound of $2$ to settings where the elements picked must form a matroid. This bound is tight in the sense that it comprises the case where both the online and offline algorithm have to pick one number as a special case, for which this bound is known to be tight. 

\citet{HaKl07} observed the following relationship between prophet inequalities and algorithmic mechanism design: algorithms used to prove prophet inequalities can be interpreted as truthful online auction mechanisms, and the prophet inequality in turn can be interpreted as the mechanism's approximation guarantee. \citet{ChHa10} observed an even subtler relationship between the two topics: questions about the approximability of offline Bayesian optimal mechanisms by sequential posted-price mechanisms could be translated into questions about prophet inequalities, via the use of virtual valuation functions. \citet{Al11} and \citet{KlWe12}, armed with stronger prophet inequalities, deepen this relationship even further.

Another related line of literature is work on secretary problems, which also concerns relations between optimal offline stopping rules and suboptimal online stopping rules, but under the assumption of a randomly ordered input rather than independent random numbers in a fixed order. While the polymatroid prophet inequality that we solve here contains the matroid prophet inequality problem as a special case, the matroid secretary problem introduced by \citet{BaIm07} remains largely unsolved despite recent progress.

A final related direction is work on exponential-sized Markov decision processes (MDP's) \cite{DeGo04,GuMu09,GuMu10}. The connection here is that algorithms for prophet inequalities can be formulated as exponential-sized MDP's, whose state reflects the entire set of decisions made prior to a specified point during the algorithm's execution. Most of the algorithms with provable approximation guarantees for exponential-sized MDP's are LP-based, while our algorithm is combinatorial.

%% file: pm-prelims.tex
\section{Preliminaries} \label{sec:preliminaries}

\paragraph{Bayesian Online Selection Problems}
In a \emph{Bayesian online selection problem} we are given a ground set $\U$ and for each $x \in \U$ a probability distribution $F_x$ with support $\reals_+$. This induces a probability distribution over functions $w: \U \rightarrow \reals_+$ in which the random variables $\{w(x) \midd x \in \U\}$ are independent and $w(x)$ has distribution $F_x$. We refer to $w(x)$ as the \emph{weight} of $x$. The goal is to choose a vector $z \in \reals^{|\U|}$ that maximizes $w \cdot z = \sum_{x \in \U} w(x) \cdot z(x)$. For a given assignment of weights we use $\OPT(w)$, or simply $\OPT$, to denote the optimal value. The vector $z$ will typically be restricted to come from a space of feasible vectors $\F \subseteq \reals^{|\U|}.$ One common restriction is $\F \subseteq \{0,1\}^{|\U|}$ in which case $z_i \in \{0,1\}$ can be thought of as encoding membership to a subset $A \subseteq \U$. Two further restrictions, matroids and polymatroids, are discussed below.

An \emph{input sequence} is a sequence $\sigma$ of ordered pairs $(x_i,w_i)_{i=1,..|\U|}$ such that $(x_i,w_i) \in \U \times \reals_+$ and every element of $\U$ occurs exactly once in the sequence. A \emph{deterministic online selection algorithm} is a function $z$ mapping every input sequence $\sigma$ to a vector $z(\sigma) \in \F$ such that for any pair of input sequences $\sigma, \sigma'$ that match on the first $i$ pairs $(x_1,w_1), \dots, (x_i,w_i)$ we have $z_j(\sigma) = z_j(\sigma')$ for all $1 \le j \le i.$ An \emph{online-weight adaptive adversary} that has chosen $x_1\dots,x_i$ and has learned about $w(x_1), \dots, w(x_{i-1})$ chooses $x_i$ without knowing $w(x_i)$.

\paragraph{Matroids} 
A \emph{matroid} $\M$ is a pair $(\U,\I)$, where $\U$ is a set (called the \emph{ground set}) and $\I \subseteq 2^\U$ is a non-empty, downward-closed family of subsets of $\U$ (called the \emph{independent sets}) satisfying the matroid exchange axiom: for all pairs of sets $I, J \in \I$ such that $|I| < |J|$ there exists an element $x \in J$ such that $I \cup \{x\} \in \I.$ A maximal independent set $I \in \I$ is called a \emph{basis}.

\paragraph{Polymatroids} 
A \emph{polymatroid} $\Pf$ on ground set $\U$ is given by $\Pf = \{z \in \reals^{|\U|} \midd z(T) \le f(T) \text{ for all } T \subseteq \U \}$, where $z(T) = \sum_{i \in T} z_i$ and $f: 2^\U \rightarrow \reals$ is a submodular set function. A set function $f$ is \emph{submodular} if for every two sets $X \subset Y \subseteq \U$,
\[
	f(X \cap Y) + f(X \cup Y) \le f(X) + f(Y).
\]
A submodular function is integer-valued if for every subset $X \subseteq \U$, $f(X) \in \naturals.$

\paragraph{Notation}
For a real number $z$, we use $z^+$ to denote $\max\{z,0\}$. For a vector $w$ with entries indexed by a set $\U$, and for any set $S \subseteq \U$, we use $w(S)$ to denote $\sum_{i \in S} w_i$.

%% file: pm-alg.tex
\section{Algorithm for Polymatroids}

Our algorithm for the polymatroid prophet inequality is based on the algorithm of Kleinberg and Weinberg~\cite{KlWe12} for the matroid prophet inequality. We begin by defining {\em block-structured matroids, block-restricted weight distributions}, and {\em block-restricted adversaries}. The crux of our analysis is a theorem (whose proof is deferred to \secref{sec:proof} below) asserting that the Kleinberg-Weinberg algorithm, applied to block-structured matroids with a block-restricted adversary, recovers at least half of the optimal reward. Armed with this theorem, we design our algorithm for polymatroids by reducing to the block-structured matroid case; we have tailored the definition of block-restricted weight distribution and block-structured adversary so that they capture the type of input sequences generated by our reduction.

\subsection{Block-Structured Matroids}\label{sec:blocks}

We begin by defining block-structured matroids and showing that to every polymatroid defined by an integer-valued submodular function there is an associated block-structured matroid. Afterwards, we define block-restricted weight distribution and block-restricted adversary.

\begin{definition}\label{def:block-structured-matroids}
A \emph{block-structured matroid} is one whose ground set is partitioned into blocks $B_1,...,B_n$ such that the independence relation is preserved under permutations of the ground set that preserve the pieces of the partition.

For a set $S \subseteq B_1 \cup \cdots \cup B_n$ we define its cardinality vector $\cvec{S} = (\card{1}{S},\card{2}{S},\ldots,\card{n}{S})$ by setting $\card{i}{S} = |S \cap B_i|$ for $i=1,\ldots,n$.
\end{definition}

\begin{lemma}\label{lem:reduction}
Suppose $f$ is a submodular function on ground set $\U = \{u_1,\ldots,u_n\}$, taking values in $\{0,1,\ldots,M\}$. There is a block-structured matroid $\M_f$ on ground set $\U \times [M]$ with blocks $B_i = \{u_i\} \times [M] \; (i=1,\ldots,n)$, whose independent sets are those $S$ satisfying $\cvec{S} \in \Pf$.
\end{lemma}
\begin{proof}
The bulk of the proof is devoted to proving that the independent sets constitute a matroid. The criterion for a set to be independent depends only on the cardinality of its intersection with each block $B_i$, hence is clearly preserved under permutations that preserve the blocks. Thus, the fact that the matroid is block-structured will follow trivially once we have established that it is indeed a matroid.

Clearly the empty set is independent and a subset of an independent set is independent. To verify the matroid exchange axiom suppose that we have two sets $S,T \in \I$ such that $|S| < |T|$. 
Define additive set functions $x,y$ on subsets of $\U$ by
\[
x(A) = \sum_{i \in A} \card{i}{S}, \quad
y(A) = \sum_{i \in A} \card{i}{T}.
\]
By our assumption that $S,T \in \I$ we have $x(A) \leq f(A)$ and $y(A) \leq f(A)$ for all $A$. Define a set $A$ to be $x$-tight if $x(A) = f(A)$. For any two sets $A,B$ we have $x(A) + x(B) = x(A \cap B) + x(A \cup B)$, and from this it is easy to deduce that the union and intersection of $x$-tight sets is $x$-tight. In particular, the set of all elements that belong to $x$-tight sets, is itself an $x$-tight set. Denote that set by $R$. We have $y(R) \leq f(R) = x(R)$. On the other hand, for the ground set $\U$ we have $y(\U) = |T| > |S| = x(\U)$. Hence, there must be some element $i \not\in R$ such that $y(\{i\}) > x(\{i\})$. Let $z$ be any element of $T \cap B_i$ that does not belong to $S \cap B_i$. The set $S \cup \{z\}$ is an independent set in the matroid, because $i$ does not belong to any $x$-tight set and hence the vector $x$ remains in the polymatroid after incrementing its $i^{\mathrm{th}}$ coordinate. This verifies the matroid exchange axiom.
\end{proof}


\begin{definition}\label{def:block-restricted-weight-distribution}
A \emph{block-restricted weight distribution} on a block-structured matroid is a joint distribution of weights for its elements, such that the elements of a block receive identical weights, and the weight assignments to different blocks are mutually independent.
\end{definition}

\begin{definition}\label{def:block-restricted-adversary}
A \emph{block-restricted adversary} is one who is restricted to choose an ordering of the input sequence in which the elements of each block appear consecutively, and after any proper subset of the blocks have been presented, the choice of which block is presented next may only depend on the weights of elements that have already been presented.
\end{definition}

Note that when all blocks have size 1, a block-structured matroid is simply a matroid, and a block-restricted distribution is simply an independent distribution. Furthermore, a block-restricted adversary is exactly the same as the notion of \emph{online weight-adaptive adversary} defined in~\cite{KlWe12}. Thus, the special case in which all blocks have size 1 is precisely the setting of the matroid prophet inequality of~\cite{KlWe12}. 

\subsection{Prophet Inequality for Block-Restricted Distributions and Adversaries}\label{sec:matroids}

Consider a block-restricted matroid $(\U,\I)$. Let $w,w': \U \rightarrow \reals_+$ denote two assignments of weights to the elements of $\U$ sampled independently from a block-restricted weight distribution. For a given input sequence $\sigma = (x_1,w(x_1)), \dots, (x_n,w(x_n))$ we compare the set $A = A(\sigma)$ selected by the algorithm to the basis $B$ that maximizes $w'(B)$. 

The matroid exchange axiom guarantees the existence of a partition of $B$ into disjoint subsets $C, R$ such that $A \cup R$ is also a basis of $\M$. Among all such partitions, let $C(A), R(A)$ denote the one that maximizes $w'(R).$ Let $g(A) = w'(R(A))$.

The selection algorithm is as follows: In step $i$, having already selected the (possibly empty) set $A_{i-1}$, we set threshold $T_i = \infty$ if $A_{i-1} \cup \{x_i\} \not\in \I$, and otherwise
\begin{align}
	T_i 
	&= \frac{1}{2} \cdot \E[g(A_{i-1}) - g(A_{i-1} \cup \{x_i\})] \nonumber \\
        &= \frac{1}{2} \cdot \E[w'(R(A_{i-1})) - w'(R(A_{i-1} \cup \{x_i\}))]. \label{eq:t-1} \\
	&= \frac{1}{2} \cdot \E[w'(C(A_{i-1} \cup \{x_i\})) - w'(C(A_{i-1}))]. \label{eq:t-2}
\end{align}
We select element $x_i$ if and only if $w_i \ge T_i.$

\begin{theorem}\label{thm:matroid}
For every block-restricted matroid $(\U,\I)$ with block-restricted weight distribution there is a deterministic online selection algorithm that achieves the following performance guarantee against block-restricted adversaries:
\[
	\E[w(A)] \ge \frac{1}{2} \cdot \OPT.
\]
\end{theorem}
Before providing a proof of this theorem in \secref{sec:proof}, we show how it can be used to derive a prophet inequality for polymatroids.

\subsection{A Prophet Inequality for Polymatroids}\label{sec:polymatroids}

The algorithm that achieves the prophet inequality in the polymatroid setting (with rational-valued submodular function $f$) does so by reducing the problem to the block-structured matroid setting with the matroid $\M_f$ defined in Lemma~\ref{lem:reduction}.

If in the polymatroid setting the elements are presented in order $u_1,\dots,u_n$, then the reduction constructs an input sequence in the matroid setting by presenting the elements in order $(u_1,1), (u_1,2),\dots, (u_2,1), (u_2,2), \dots$ (lexicographic order, $\U$ coordinate first). If in the polymatroid setting the weight of element $u_i$ is $w_i$ then element $(u_i,j)$ is presented in the matroid setting with weight $w_i$. If the matroid algorithm, while processing elements $(u_i,1),(u_i,2),\ldots,(u_i,M)$, selects a subset $\{u_i\} \times S_i$, then the polymatroid algorithm when processing $u_i$ sets $z_i = |S_i|.$

\begin{theorem}
For every polymatroid $\Pf$ defined by a rational-valued submodular function $f$ there exists a deterministic online selection algorithm that satisfies the following performance guarantee against online weight-adaptive adversaries:
\[
	\E \left[ \sum_i w_i \cdot z_i \right] \ge \frac{1}{2} \cdot \OPT.
\]
\end{theorem}
\begin{proof}
Assume w.l.o.g.\ that $f$ is integer-valued; the extension to rational-valued functions follows by a trivial scaling argument. The matroid $\M_f$ is block-structured, and the weights of the elements $(u_i,j)$ generated by the reduction are sampled from a block-restricted distribution since the weights $w_1,\ldots,w_n$ are mutually independent random variables, and the weights in block $B_i$ are all equal to $w_i$.
Furthermore, the method for constructing the input sequence satisfies the definition of a block-restricted adversary, since the elements of each block appear consecutively.

For any cardinality vector $\mathbf{q}$, the total weight of any set $S \subseteq \U \times [M]$ such that $\cvec{S} = \mathbf{q}$ is equal to $w \cdot \mathbf{q}$. In particular, this implies: 
\begin{enumerate}
	\item The value of the vector $\mathbf{z}$ selected by our polymatroid algorithm is equal to the sum of weights of elements selected in its internal simulation of the matroid algorithm.
	\item The value of $\OPT$ in the matroid setting is equal to $\max \{ \sum_{i=1}^n w_i \cdot q_i \mid \mathbf{q} \in \Pf\}$, which coincides with $\OPT$ in the polymatroid setting.
\end{enumerate}
Combining these two facts with~\thmref{thm:matroid}, we obtain the performance guarantee in the theorem statement.
\end{proof}

\subsection{Proof of the Block-Restricted Matroid Prophet Inequality}\label{sec:proof}

We start with a proposition that generalizes the corresponding result of \citet{KlWe12} from independent weight distributions to block-restricted weight distributions. The proof is straightforward and given in \appref{app:balanced-pt-1}. 

\begin{proposition}\label{pro:balanced-pt-1}
For every input sequence $\sigma$, if $A = A(\sigma)$, then
\[
	\sum_{x_i \in A} T_i = \frac{1}{2} \cdot \E[w'(C(A))].
\]
\end{proposition}

\begin{figure}[ht]
  \center
  \begin{tikzpicture}
    \draw[->, semithick] (0,0) -- (8,0) node[right = 0.1cm] {$x$};
    \draw[->, semithick] (0,0) -- (0,4);
    \draw[-, semithick, dotted] (8,2) -- (0,2) node[left = 0.1cm] {$w$};
    \foreach \i in {0,1,2,3,4,5,6,7} {
      \node[draw,circle,inner sep=1pt,fill] at (\i,3-3/2^\i) {};
    }
    \foreach \i in {0,1,2} {
      \node[draw,circle,inner sep=2.5pt] at (\i,3-3/2^\i) {};
    }
    \foreach \i in {3,4,5,6,7} {
      \node[draw,circle,inner sep=2.5pt] at (\i,3-3/2^2) {};
    }
    \draw[-] (0,0) -- (0,0) node[below = 0.1cm] {$i_0$};
    \draw[-] (7,0) -- (7,0) node[below = 0.1cm] {$i_1-1$};
    \draw[-] (4,0) -- (4,0) node[below = 0.1cm] {$i$};
    \draw[-] (4,3-3/2^4) -- (4,3-3/2^4) node[above = 0.1cm] {$t_i$};
    \draw[-] (4,3-3/2^2) -- (4,3-3/2^2) node[below = 0.3cm] {$T_i$};
  \end{tikzpicture}
  \caption{Visualization of the thresholds set by the algorithm}
  \label{fig:thresholds}
\end{figure}
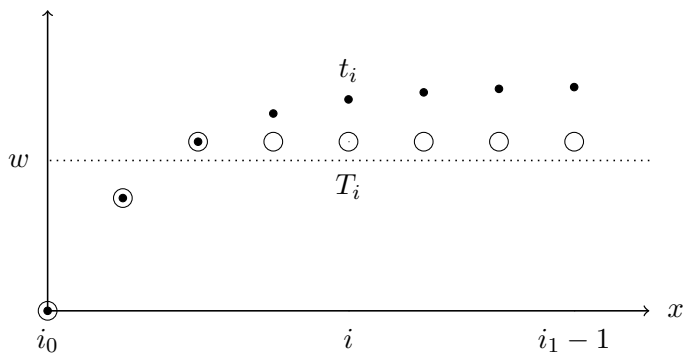

Next we prove that the thresholds within a given block have a specific form (see \figref{fig:thresholds} for an illustration). Specifically, consider any block consisting of elements $u_{i_0}, u_{i_0+1}, \dots, u_{i_1-1}$. For all $i_0 \le i \le i_1$ define $A^i = A_{i_0-1} \cup \{x_{i_0},\ldots,x_i\}$, and
\begin{equation} \label{eq:ti}
	t_i = \frac{1}{2} \cdot \E_{w'}\left[g(A^{i-1})-g(A^i)\right],
\end{equation}
where for convenience we also set $A^{i_0-1} = A_{i_0-1}$. We will show that the sequence of numbers defined by~\eqref{eq:ti} forms a non-decreasing sequence depending only on the weights associated with previous elements $w_{1}, w_{2}, \dots, w_{i_0-1}$, and that for $i_0 \le i \le i_1-1$ the algorithm sets threshold $T_i = t_i$ if $t_i \le w$ and $T_i > w$ otherwise. 

\begin{lemma}\label{lem:thresholds}
Consider a block-structured matroid $(\U,\I)$ with blocks $B_1,\ldots,B_n$. For any input sequence $\sigma$ generated by a block-restricted adversary, and any block $B_j$, let $i_0, i_0+1, \ldots, i_1$ denote the times when the elements of $B_j$ are presented in $\sigma$. The sequence of numbers $t_{i_0}, \ldots, t_{i_1}$ defined by~\eqref{eq:ti} satisfies $t_{i_0} \leq t_{i_0+1} \leq \dots \leq t_{i_1}$ and depends only on the subsequence of $\sigma$ preceding time $i_0$. Moreover, the algorithm sets $T_i = t_i$ for all $i_0 \le i \le i_1$ such that $t_i \le w_i$, and $T_i > w_i$ otherwise.
\end{lemma}
\begin{proof}
By definition $t_i$ depends only on $A_{i_0-1}$ and $i$, thus, only on the subsequence of $\sigma$ preceding time $i_0$.

To prove that $t_i \le t_{i+1}$ for $i_0 \le i < i_1$ we will show that the inequality 
\begin{equation} \label{eq:gdiff}
g(A^{i-1}) - g(A^i) \geq g(A^i) - g(A^{i+1})
\end{equation}
holds pointwise (\ie for every choice of $w'$) and not just in expectation. Recall that $g(S) = w'(R(S))$ and note that $A^{i-1} \subset A^{i} \subset A^{i+1}$, with each set in the chain containing one more element of $B_j$ than the preceding one.

Now consider that for any independent set $S$, the set $R(S)$ is formed by going through the elements of $\M$ in decreasing order of $w'$, selecting every element that is not spanned\footnote{In a matroid, we say that $x$ is spanned by $T$ if $T \cup \{x\}$ has a maximal independent set that is disjoint from $\{x\}$.} by the union of $S$ with the earlier elements in the list. Consequently, for any $S \in \I$ and $x \not\in S$, we have $R(S \cup \{x\}) \subset R(S)$ and the unique element of $R(S) \setminus R(S \cup \{x\})$ is the first element that is spanned by earlier elements combined with $S \cup \{x\}$ but not $S$; let us call this the \emph{critical element} for $(S,x)$.
Let $S = A^{i-1}, \, S \cup \{x\} = A^{i}, \, S \cup \{x,y\} = A^{i+1}$. We find that the first time an element is spanned by earlier elements combined with $S \cup \{x\}$ it is also spanned by earlier elements combined with $S \cup \{x,y\}$, and consequently the critical element for $(S \cup \{x\}, y)$ occurs in the same place or earlier than the critical element for $(S, x)$. Consequently, the critical element for $(S \cup \{x\}, y)$ has the same or greater weight than the critical element for $(S,x)$, \ie
$$ g(A^{i}) - g(A^{i+1}) \geq g(A^{i-1}) - g(A^{i})$$ 
as desired.

Having proven that $t_i$ is monotonically non-decreasing in $i$, we shall now prove that for all $i_0 \le i \le i_1$ such that $t_i \le w_i$, the algorithm sets $T_i = t_i$ and selects $i$. The proof is by induction on $i$. From the definition of $T_i$ and $t_i$, it is clear that $T_i = t_i$ provided that the algorithm has selected $i_0,\ldots,i-1$. Thus $T_{i_0}=t_{i_0}$ (the base case) and for $i_0 < i \le i_1$ such that $t_i \le w_i$ the induction hypothesis implies that the algorithm has already selected $i_0,\ldots,i-1$ thus establishing $T_i = t_i$. Now by the algorithm's selection criterion, the relation $T_i = t_i \leq w_i$ implies then $i$ is selected, which concludes the proof of the base case and induction step.

To conclude the proof of the lemma, we consider the case of $i$ such that $t_i > w_i$. Let $i_2$ denote the least such $i$. We will prove, again by induction on $i$, that $T_i = T_{i_2} = t_{i_2} > w_{i_2} = w_i$ for all such $i$. Since the algorithm has selected elements $i_0,\ldots,i_2-1$, we have $T_{i_2} = t_{i_2}$ which establishes the base case. For the induction step, note that the induction hypothesis implies that the algorithm does not select any elements in the range $i_2,\ldots,i-1$. Consequently $T_i = \frac12 \cdot \E[g(A^{i_2-1}) - g(A^{i_2-1} \cup \{x_i\})]$. The relation $A^{i_2-1} \cup \{x_i\} = A^{i_2-1} \cup \{x_{i_2}\}$ implies that $g(A^{i_2-1} \cup \{x_i\}) = g(A^{i_2-1} \cup \{x_{i_2}\})$ and hence $T_i = T_{i_2}$ which completes the induction step.
\end{proof}

An important corollary of the preceding structural result regarding the thresholds is the following assertion for two weight assignments $w, w'$ drawn independently from a block-restricted weight distribution.

\begin{corollary}\label{cor:fix}
Let $w, w'$ be two weight assignments drawn independently from a block-restricted weight distribution. For any input sequence $\sigma$ generated by a block-restricted adversary, and any block $B_j$, let $i_0, i_0+1, \ldots, i_1$ denote the times when the elements of $B_j$ are presented in $\sigma$. 
Then, for all $i_0 \le i < i_1$, $(w_i-T_i)^+ = (w_i-t_i)^+$, and $w_i, t_i, w'(x_i)$ are mutually independent, so
\[
	\E[(w_i-T_i)^+] = \E[(w_i-t_i)^+] = \E[(w'(x_i)-t_i)^+]. 
\]
\end{corollary}

Finally, before proving \thmref{thm:matroid} we need to prove an inequality analogous to Proposition 2 of \citep{KlWe12}, but using the surrogate thresholds $t_i$ in place of the algorithm's actual thresholds $T_i$.

\begin{proposition}\label{pro:balanced-pt-2}
For every input sequence $\sigma$ generated by a block-restricted adversary, let $A = A(\sigma)$, and let $R'(A)$ be a set such that $\cvec{R'(A)} = \cvec{R(A)}$ and $R'(A)$ contains the earliest $|R(A) \cap B_j|$ elements of each block $B_j \; (1 \le j \le n)$.
Then 
\[ 
       \sum_{x_i \in R'(A)} t_i \le \frac12 \E[w'(R'(A))] = \frac12 \E[w'(R(A))].
\]
\end{proposition} 
\begin{proof}
Recalling the definition of $t_i$ in~\eqref{eq:ti}, we see that it suffices to prove that the following holds pointwise (\ie for every choice of $w'$).
\begin{equation} \label{eq:bal-pt-2.0}
      \sum_{x_i \in R'(A)} g(A^{i-1}) - g(A^i) \le w'(R'(A)) = w'(R(A)).
\end{equation}
The equation $w'(R'(A)) = w'(R(A))$ is an immediate consequence of the fact that $\cvec{R'(A)} = \cvec{R(A)}$ and $w'(S) = w' \cdot \cvec{S}$ for any set $S$.

Let $\wpm$ be a permutation of $\M$ that preserves each block and maps $R(A)$ to $R'(A)$. Note that $A \cup R(A) \in \I$ so $\wpm(A) \cup R'(A) \in \I$ as well, since $\M$ is block-structured. To bound the left side of~\eqref{eq:bal-pt-2.0}, we break up the sum into separate sums, one for each block of $\M$. 
For $j=1,\ldots,n$ let $i_0(j)$ denote the initial index of block $B_j$, and let $R_j = R(A) \cap B_j, \; R'_j = R'(A) \cap B_j$. We have
\begin{align}
\nonumber
      \sum_{j=1}^n \sum_{x_i \in R'_j} g(A^{i-1}) - g(A^i)
  & =
      \sum_{j=1}^n g(A^{i_0(j)}) - g(A^{i_0(j)} \cup R'_j) \\
\nonumber  & \leq
      \sum_{j=1}^n g(\wpm(A)) - g(\wpm(A) \cup R'_j) \\
\nonumber  & =
      \sum_{j=1}^n w'(R(\wpm(A))) - w'(R(\wpm(A) \cup R'_j)) \\
\label{eq:bal-pt-2.1}  & = 
      \sum_{j=1}^n w'(R(A)) - w'(R(A \cup R_j)).
\end{align}
The second line follows from the fact that the restriction of $g$ to the independent set $\wpm(A) \cup R'(A)$ is submodular (Lemma 3 of \citep{KlWe12}) and $\wpm(A)$ is disjoint from $R'(A)$. The last line follows from the preceding one by applying the weight-preserving matroid automorphism $\wpm^{-1}$ to all the sets involved.

Observe that $R(A \cup R_j) = R(A) \setminus R_j$. This is because $R(A \cup R_j)$ is the maximum-weight subset $R \subseteq B$ such that $A \cup R_j \cup R$ is a basis of $\M$, and $R(A) \setminus R_j$ is one such subset. Furthermore, if there were any other set $R$ such that $A \cup R_j \cup R$ were a matroid basis and $w'(R) > w'(R(A) \setminus R_j)$ then $R_j \cup R$ would have greater weight than $R(A)$, contradicting the definition of $R(A)$.
Plugging the relation $R(A \cup R_j) = R(A) \setminus R_j$ into the right side of~\eqref{eq:bal-pt-2.1}, we find that
\[
	\sum_{j=1}^n w'(R(A)) - w'(R(A \cup R_j))
	= \sum_{j=1}^n w'(R_j)
	= w' \left( \bigcup_{j=1}^n R_j \right)
	= w'(R(A)),
\]
which completes the proof of~\eqref{eq:bal-pt-2.0} and hence of the entire proposition.
\end{proof}

We are now ready to prove \thmref{thm:matroid}.

\begin{proof}[Proof of \thmref{thm:matroid}]
Since $C(A) \cup R(A)$ is a maximum-weight basis with respect to $w'$, and $w'$ and $w$ are identically distributed we have
\[
	\OPT = \E[w'(C(A)) + w'(R(A))].
\]
Recalling the set $R'(A)$ from \proref{pro:balanced-pt-2}, we will derive the following inequalities:
\begin{align}
	&\E\left[\sum_{x_i \in A} T_i\right] \ge \frac{1}{2} \cdot \E[w'(C(A))], &&\text{and} \label{eq:a}\\
	&\E\left[\sum_{x_i \in A} (w(x_i)-T_i)^+\right] \ge \E\left[\sum_{x_i \in R'(A)} (w'(x_i) - t_i)^+\right], &&\text{and} \label{eq:b}\\
	&\E\left[\sum_{x_i \in R'(A)} (w'(x_i)-t_i)^+\right] \ge \frac{1}{2} \cdot \E\left[w'(R(A))\right]. \label{eq:c}
\end{align} 
By adding inequalities (\ref{eq:b}) and (\ref{eq:c}) to inequality (\ref{eq:a}) and using the fact that $T_i + (w(x_i)-T_i)^+ = w(x_i)$ for all $x_i \in A$, we obtain  
\[
	\E[w(A)] \ge \frac{1}{2} \cdot \E[w'(C(A))] + \frac{1}{2} \cdot \E[w'(R(A))].
\]

Inequality (\ref{eq:a}) follows from \proref{pro:balanced-pt-1}.
For inequality (\ref{eq:b}) we use \corref{cor:fix} and that the algorithm picks every $i$ such that $w(x_i) > T_i$ to obtain
\[
	\E\left[ \sum_{i \in A} (w(x_i) - T_i)^+ \right] = \E\left[\sum_{i=1}^{n}(w(x_i)-T_i)^+\right] = \E\left[\sum_{i=1}^{n} (w'(x_i) - t_i)^+\right] \ge \E\left[\sum_{x_i \in R'(A)}(w'(x_i) - t_i)^+\right]. \label{eq:breaks}
\]
For inequality (\ref{eq:c}) we apply \proref{pro:balanced-pt-2} to obtain
\begin{align*}
	\E\left[\sum_{x_i \in R(A)} w'(x_i)\right] 
        &= \E\left[\sum_{x_i \in R'(A)} w'(x_i) \right] \\
	&\leq \E\left[\sum_{x_i \in R'(A)} t_i\right] + \E\left[\sum_{x_i \in R'(A)} (w'(x_i) - t_i)^+\right]\\
	&\leq \frac{1}{2} \cdot \E\left[\sum_{x_i \in R(A)} w'(x_i)\right] + \E\left[\sum_{x_i \in R'(A)} (w'(x_i) - t_i)^+\right].
\end{align*}
Adding $-\frac12 \cdot \E[\sum_{x_i \in R(A)} w'(x_i)]$ on both sides finishes the proof.
\end{proof}

%% file: pm-md.tex
\section{Applications in Mechanism Design}\label{sec:applications}

Prophet inequalities are an important tool for the design of simple yet approximately optimal mechanisms \citep{ChHa10}. Previously known prophet inequalities could not be applied to settings in which the set of feasible solutions forms a polymatroid. Our prophet inequality thus leads to improved results for these settings.



In a {\em polymatroid single-parameter Bayesian mechanism design problem} we are given a set $\U$ of $n$ agents that strive to be serviced. Each agent $i$ has a private value $v_i\in \reals^+$ for being serviced. The value $v_i$ of agent $i$ is drawn from the cumulative distribution function $F_i$. A mechanism $(x,p)$ consists of an outcome rule $x: \reals_+^n \rightarrow \reals_+^n$, where $x_i$ specifies how much service agent $i$ gets, and a payment rule $p: \reals_+^n \rightarrow \reals_+^n$, where $p_i$ specifies the payment of agent $i$. An outcome is feasible if $\sum_{i \in S} x_i \le f(S)$ for all $S \subseteq \U$, where $f$ is a submodular function. If $f$ is rational valued we say that the problem has rational constraints. An agent's utility is linear in the quantity of service it receives and its payment. That is, agent $i$'s utility is $u_i(b,v_i) = v_i \cdot x_i(b) - p_i(b)$, where $b$ denotes the bids of the agents. The social welfare is $\sum_{i \in \U} v_i \cdot x_i(b)$ and the revenue is $\sum_{i \in \U} p_i(b)$. A mechanism is dominant strategy incentive compatible if for every agent $i$, value $v_i$, bid $b_i$ and bids $b_{-i}$, $u_i((v_i,b_{-i}),v_i) \ge u_i((b_i,b_{-i}),v_i).$

A number of polymatroid single-parameter Bayesian mechanism design problems are given in \citet{BidV11}. The following two are from \citet{GoMi12} and \citet{BaNi09} and can be used to model sponsored search and video on demand. 

\begin{itemize}
\item {\bf Position Auctions:} There are $n$ agents and $m$ instances. Each advertiser $i$ is interested in a subset of instances $\Gamma(i) \subseteq [m]$. For each instance $k$ let $\Gamma(k)$ denote the agents that are interested in it. Each instance is associated with $|\Gamma(k)|$ positions. Position $j$ for instance $k$ has quality $\alpha_j^k$ such that $\alpha_1^k \ge \alpha_2^k \ge \dots \ge \alpha_{|\Gamma(k)|}^k$. Let $\mathcal{A}_k = \{ \pi_k: \Gamma(k) \rightarrow [|\Gamma(k)|]\}$ denote the set of allocations (one to one maps) of agents to instances. Let $\Delta(\mathcal{A}_k)$ denote the distributions over such allocations. Allocation $x=(x_1,\dots,x_n)$ is feasible if there is a distribution over allocations of agents to positions for each instance such that agent $i$ gets $x_i$ clicks in expectation. This is a polymatroid \cite{McDi75}.
\item {\bf Spatial Markets:} There are $n$ agents and one seller. There is a capacitated network in which each edge has a capacity. The agents correspond to disjoint sets of demand nodes. The seller corresponds to a source node. The agents are interested in the sum of flows $x_i$ into their nodes $i$. A solution $x$ is feasible if and only if $\sum_{e \in S} x_e \le f(S)$ for all $S$, where $f(S)$ is the value of a minimum $s$-$S$-cut. This is a polymatroid \citep{FeGr86}. 
\end{itemize}

These problems have rational constraints if the qualities and capacities are rational-valued, which is a reasonable assumption in the sponsored search and video on demand application. In the former the qualities correspond to clicks in the latter the capacities correspond to the number of videos that can be simultaneously streamed.


By applying the sequential posted pricing technique of \citet{ChHa10}, and using our prophet inequality for polymatroids we obtain simple dominant-strategy incentive compatible mechanisms that are guaranteed to achieve at least half of the optimal revenue. For the position auctions problem this is better than the best known bound for another simple mechanism, Generalized Second Price (GSP) with reserve prices \cite{LuPa12,CaKa12}. For the spatial markets problem we are not aware of any approximation results.

\begin{theorem}
For polymatroid single-parameter Bayesian mechanism design problems with rational constraints the revenue obtained by sequential posted pricing is within a factor of two of the optimal revenue.
\end{theorem}

%% file: pm-appendix.tex
\section{Proof of \proref{pro:balanced-pt-1}}\label{app:balanced-pt-1}

We use linearity of expectation and a telescoping sum to obtain,
\begin{align*}
	\sum_{x_i \in A} T_i 
	&= \frac{1}{2} \cdot \sum_{x_i \in A} \E[w'(C(A_{i-1} \cup \{x_i\})) - w'(C(A_{i-1}))]\\
	&= \frac{1}{2} \cdot \sum_{x_i \in A} \E[w'(C(A_i)) - w'(C(A_{i-1}))]\\
	&= \frac{1}{2} \cdot \E[w'(c(A_n)) - w'(c(A_0))]\\
	&= \frac{1}{2} \cdot \E[w'(C(A))]. \qedhere
\end{align*}